\documentclass[a4paper,UKenglish,cleveref, autoref, thm-restate]{lipics-v2021}

\usepackage{bbm}
\usepackage{soul}

\DeclareMathOperator{\vor}{Vor}

\DeclareMathOperator{\Int}{int}

\DeclareMathOperator{\cl}{\mathrm{cl}}

\newcommand{\Rd}{\mathbb{R}^{d}}

\newcommand{\bs}{\boldsymbol}

\newcommand{\ind}{\mathbbm{1}}

\newcommand{\R}{{\mathbb{R}}}

\newcommand{\E}{\mathbb{E}}

\newcommand{\cX}{{\cal{X}}}

\newcommand{\cN}{{\cal{N}}}

\newcommand{\cP}{{\cal{P}}}

\newcommand{\cB}{{\cal{B}}}

\newcommand{\logg}{\log\log}

\newcommand{\eps}{\epsilon}

\newcommand{\set}[1]{\left\{ #1 \right\}}

\newcommand{\splitb}{\begin{split}}
\newcommand{\splite}{\end{split}}

\newcommand{\vsimp}{{V_{\mathrm{simp}}}}

\newcommand{\cI}{\mathcal{I}}
\newcommand{\headermath}[1]{\texorpdfstring{$#1$}{TEXT}}

\newcommand{\Sp}{\mathbb{S}}

\title{Morse Theory for the k-NN Distance Function} 

\author{Yohai Reani}{Viterbi Faculty of Electrical \& Computer Engineering, Technion - Israel Institute of Technology, Haifa, Israel \and \url{https://sites.google.com/view/yohaireani/home}}{syohai@campus.technion.ac.il}{https://orcid.org/0000-0002-0615-5789}{}
\author{Omer Bobrowski}
{
Viterbi Faculty of Electrical \& Computer Engineering, Technion - Israel Institute of Technology, Haifa, Israel \and The School of Mathematical Sciences, Queen Mary University of London, E14NS London, UK
\and \url{https://sites.google.com/site/omerbobrowski/}
}
{omer@ee.technion.ac.il}{https://orcid.org/0000-0002-0860-7099}
{}

\Copyright{Yohai Reani and Omer Bobrowski} 


\authorrunning{Y. Reani, and O. Bobrowski}

\acknowledgements{The authors are grateful to Primoz Skraba, for his feedback and advice. We would also like to thank the anonymous referees for their useful comments and suggestions.}

\nolinenumbers 

\begin{document}

\maketitle

\begin{abstract}
We study the $k$-th nearest neighbor distance function from a finite point-set  in $\R^d$. We provide a Morse theoretic framework to analyze the sub-level set topology. In particular, we present a simple combinatorial-geometric characterization for critical points and their indices, along  with detailed information about the possible changes in homology at the critical levels. We conclude by computing the expected number of critical points for a homogeneous Poisson process.
Our results deliver significant insights and tools for the analysis of persistent homology in order-$k$ Delaunay mosaics, and random $k$-fold coverage.
\end{abstract}

\section{Introduction}
Let $\cP$ be finite subset of $\R^d$, with $|\cP|\ge k$. We define the $k$-nearest neighbor distance ($k$-NN) function $d^{(k)}_{\cP}:\R^d\to \R^+$  as 
\[
d_{\cP}^{(k)}(x) := \min\set{r : |B_r(x)\cap \cP| \ge k},
\]
where $B_r(x)$ is a closed ball of radius $r$ centered at $x$. For $k=1$ we have the simple case of the distance function
\[
d_{\cP}^{(1)}(x) = d_{\cP}(x) := \min_{p\in\cP} \|x-p\|.
\]
The $k$-NN distance function arises naturally in numerous applications, including coverage in sensor networks, shape reconstruction, and clustering \cite{dudani1976distance,shamos1975closest}. A key reason for the interest in $d^{(k)}_{\cP}$ comes from fact that its sub-level sets are the $k$-fold covers, i.e.,
\[
(d_{\cP}^{(k)})^{-1}((-\infty, r]) = B_r^{(k)}(\cP) :=  \set{x\in \R^d : |B_r(x)\cap \cP| \ge k}.
\]  
In other words $B_r^{(k)}(\cP)$ contains all points that are covered by at least $k$ balls of radius $r$, centered at $\cP$.  For $k=1$ we denote $B_r(\cP) := B_r^{(1)}(\cP)$, which is simply the union of the balls around $\cP$.
Our main goal in this paper is to present a simple and comprehensive Morse theory for $d_{\cP}^{(k)}$, which is key to future study of this function within the context of applied and stochastic topology.

\emph{Morse theory} \cite{milnor_morse_1963} lies at the intersection of topology and analysis, linking local differential properties to global structural changes. Specifically, it analyzes how critical points of different indexes affect the homotopy type of the sub-level sets of a function. The classical definition of Morse theory applies to smooth functions, where the location and index of the critical points are determined by the gradient and Hessian, respectively. As  $d_{\cP}^{(k)}$ is not a differentiable function, the original notions do not apply anymore.

In \cite{bobrowski_distance_2014} the authors provided a combinatorial-geometric description for the critical points of the distance function $d_{\cP}$, their index and homological effect, based on an adaptation of Morse theory to min-type functions \cite{gershkovich_morse_1997}.  
The key property of $d_{\cP}$ which enabled the results in \cite{bobrowski_distance_2014} is that $d_{\cP}^2$ is a min-type function, i.e., it can be expressed as the minimum of a finite collection of differentiable functions. This property, however, does not extend to $d_{\cP}^{(k)}$ ($k>1$), rendering the previous approach inapplicable. In response, our paper adopts an alternative strategy, employing a more expansive Morse-theoretic framework  \cite{agrachev1997morse} designed for piecewise smooth functions. 
Leveraging this framework, we establish a simplified combinatorial-geometric representation of critical points and their homological effect.
Notably, this description generalizes the one  in \cite{bobrowski_distance_2014} for the distance function $d_{\cP}$. 

A key motivation for this work is the study of random $k$-fold coverage \cite{chiu_stochastic_2013,edelsbrunner2021multi,flatto_random_1977,hall_coverage_1985,janson_random_1986,moran1962random,penrose2023random}.
While the $k$-fold coverage process  has an intrinsic mathematical interest, it also has applications in numerous fields. For instance, in cellular networks, $k$-fold coverage provides redundancy that guarantees the  network robustness to antennae failures \cite{wang_coverage_2011}. In shape reconstruction, guaranteeing $k$-fold coverage is useful in the context of outliers removal \cite{edelsbrunner2021multi, sheehy2012multicover}.
Other examples include wireless communication \cite{haenggi_stochastic_2009}, stochastic optimization \cite{zhigljavsky_stochastic_2007}, topological data analysis (TDA) \cite{bobrowski_homological_2019}, immunology \cite{moran1962random}, and more \cite{athreya_coverage_2004,cuevas_boundary_2004}.

A related theoretical motivation comes from the field of stochastic topology, and specifically from the study of homological connectivity for a random $k$-fold cover. For $k=1$, the critical points of $d_{\cP}$ played a key role in analyzing the last changes in the homology of the random cover $B_r(\cP)$, as $r$ is increased. Taking $\cP_n$ to be a homogeneous Poisson process on a $d$-dimensional compact manifold, with rate $n$, it was proved \cite{bobrowski_homological_2019} that passing the threshold $r = ((\log n + (i-1)\logg n)/n)^{1/d}$, the $i$-th homology of $B_r(\cP_n)$ will remain unchanged if we further increase $r$. Additionally, a functional Poisson limit was proved \cite{bobrowski2022poisson}
for the locations and radii at which the last $i$-cycles appear. Note that for $i=d$, this analysis describes the exact moment at which $B_r(\cP_n)$ covers the manifold, and the critical points of index $d$ correspond to the last uncovered connected components.
The results presented here will  play a similar role in analyzing homological connectivity  for the random $k$-fold cover $B_r^{(k)}(\cP_n)$. In particular, this will enable a detailed theoretical analysis for $k$-fold coverage problem discussed above.

We note that the Morse theoretic framework we develop here for $d_{\cP}^{(k)}$ is tightly related to the study of the order-$k$ Delaunay mosaics \cite{edelsbrunner2021step,edelsbrunner2021multi}. These simplicial complexes, denoted $\mathrm{Del}_{k}(\cP)$, generalize the Delaunay triangulation and are  analogously constructed from the order-$k$ Voronoi tessellations \cite{edelsbrunner1985voronoi}. Similarly to the alpha shapes, the authors in \cite{edelsbrunner2021multi} define a sub-complex $\mathrm{Del}_{k}(\cP,r) \subset \mathrm{Del}_{k}(\cP)$ that has the same homotopy type as the $k$-fold cover $B_r^{(k)}(\cP)$.
Thus, these sub-complexes can serve as a proxy for computing the persistent homology of the $k$-fold cover filtration. The study in \cite{edelsbrunner2021step} identifies critical configurations in  $\mathrm{Del}_{k}(\cP)$, in the sense that once the corresponding cell enters the filtration $\mathrm{Del}_{k}(\cP,r)$, it changes the Euler characteristic, and consequently the homotopy-type. What we provide here is a Morse theoretic view on such critical configurations (`steps'), showing that they in fact originate from critical points of $d_{\cP}^{(k)}$. Additionally, we are able to classify them by their index, and to provide a detailed description for the effect these critical configurations have on the homology of $\mathrm{Del}_{k}(\cP,r)$, i.e., beyond the Euler characteristic.

\section{Main Results} \label{sec:main_res}

We start by briefly reviewing the fundamental statements for Morse theory for the distance function $d_{\cP}$ \cite{bobrowski_distance_2014}, based on Morse theory for min-type function \cite{gershkovich_morse_1997}.
The assumption (here and throughout the paper) is that the points in $\cP$ are in general position. For a point $c\in \R^d$, denote $r_c := d_{\cP}(c)$, and $\cP_c^\partial  := \partial B_{r_c}(c) \cap \cP$ (where $\partial B_r(c)$ denotes the boundary of the ball).
The point $c\in\Rd$ is critical for $d_{\cP}$ if and only if $c\in \sigma(\cP_c^\partial)$ (the open simplex spanned by $ \cP_c^\partial$). The index of $c$ in this case is $\mu_c := |\cP_c^\partial|-1$. Similarly to classical Morse theory, it follows from \cite{gershkovich_morse_1997} that every such critical point of index $\mu_c = i$ either adds a new generator to the $i$-th homology of $B_r(\cP)$, or kills a generator in  the $(i-1)$-th homology.

While the function $d_{\cP}^{(k)}$ can be defined as a minimum of a finite set of functions \eqref{eq:min_max}, $d_{\cP}^{(k)}$ is not a min-type function (for $k>1$) since the minimum is over functions that are not smooth. Therefore, we switch to the more general context of piecewise smooth functions and continuous selections, developed in \cite{agrachev1997morse}. In the following when we refer to `critical points', we mean that in the sense of \cite{agrachev1997morse} (Definition 1.1).

Let $c\in \R^d$, and denote $r_c := d^{(k)}_{\cP}(c)$. Define 
\begin{equation}\label{eqn:def_sets}
\cB_c := B_{r_c}(c),\quad \cP_c := \cB_c\cap \cP,\quad
\cP_c^{\cI} := \mathrm{int}(\cB_c)\cap \cP,\quad\text{and}\quad
\cP_c^{\partial} := \partial\cB_c \cap \cP,
\end{equation}
and correspondingly,
\begin{equation}\label{eqn:def_sizes}
N_c := |\cP_c|,\quad
N^{\cI}_c := |\cP_c^{\cI}|,\quad\text{and}\quad
N^{\partial}_c := |\cP_c^{\partial}|,
\end{equation}
so that $\cP_c = \cP_c^{\cI}\cup\cP_c^{\partial}$, and $N_c = N^{\cI}_c+ N^{\partial}_c$. Note that the definition of $d_{\cP}^{(k)}$ implies that $N^{\cI}_c < k$, and since the points are in general position, we have $N^{\partial}_c\le d+1$. For examples, see Figure \ref{fig:crit_pts}.

The following theorems are the main contribution of our paper, namely the characterization of critical points and their indexes, and the changes in homology induced by critical points.

\begin{theorem}\label{thm:critical_points}
A point $c\in \R^d$ is a critical point of $d_{\cP}^{(k)}$, if and only if
$c\in\sigma(\cP_c^{\partial})$. The \emph{index} of $c$ is defined as $\mu_c:=N_c - k$. All critical points of $d_{\cP}^{(k)}$ are non-degenerate.
\end{theorem}

Since $N_c^\partial \le d+1$, and $N_c^{\cI} \le k-1$, we have $\mu_c \le d$, as expected. Additionally, in the special case $\mu_c = d$, we have only one option --   $N_c^\partial=d+1$, and $N_c^{\cI} = k-1$.
Finally, note that for $k=1$, the characterization in Theorem \ref{thm:critical_points} coincides with that of the distance function $d_{\cP}$ discussed above. For examples of critical points of $d_{\cP}^{(2)}$, see Figure \ref{fig:crit_pts}. 

\begin{figure}[ht]
     \centering
     \begin{subfigure}[b]{0.45\textwidth}
         \centering
         \includegraphics[width=0.7\textwidth]{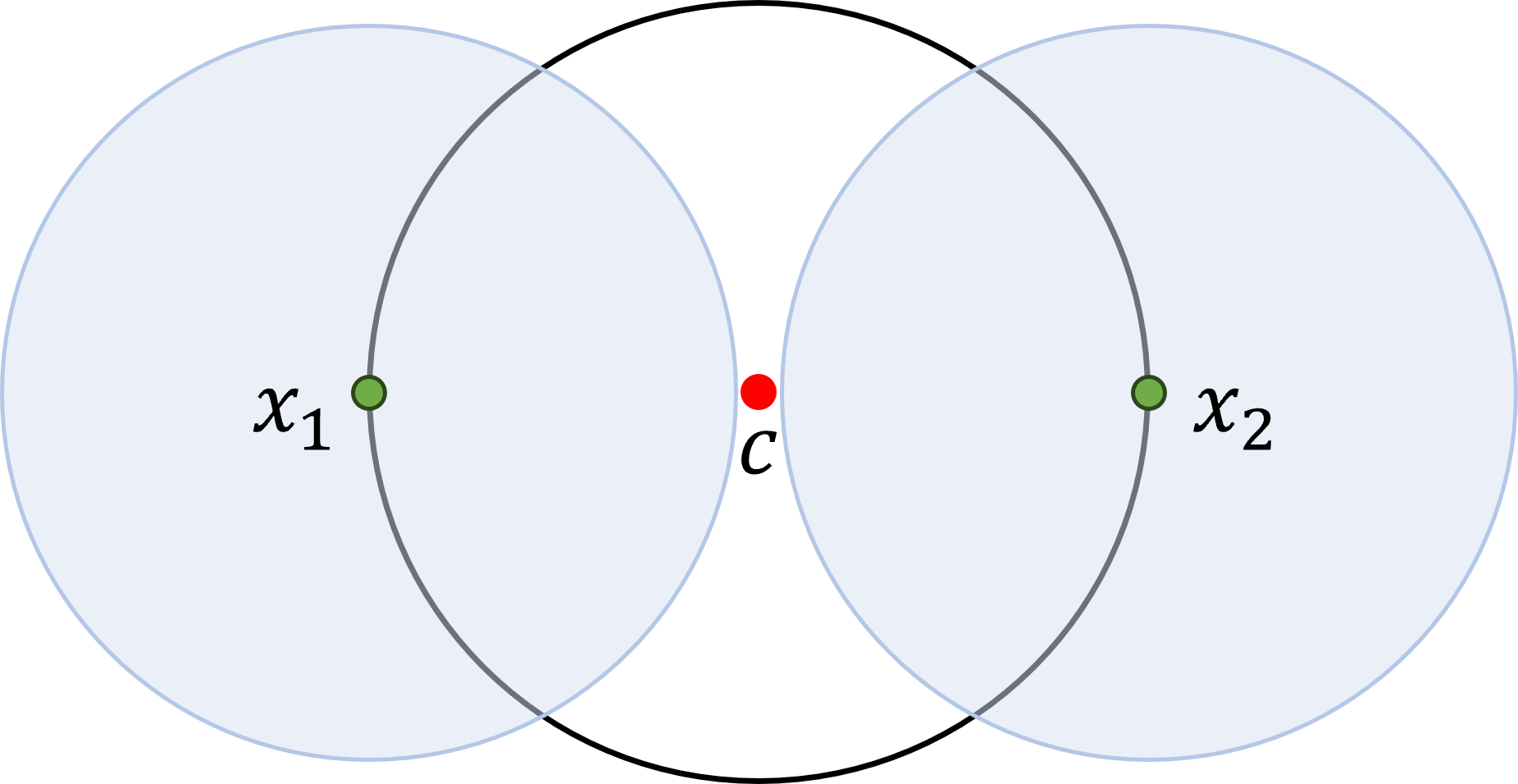}
     \end{subfigure}
     \hfill
     \begin{subfigure}[b]{0.45\textwidth}
         \centering
         \includegraphics[width=0.7\textwidth]{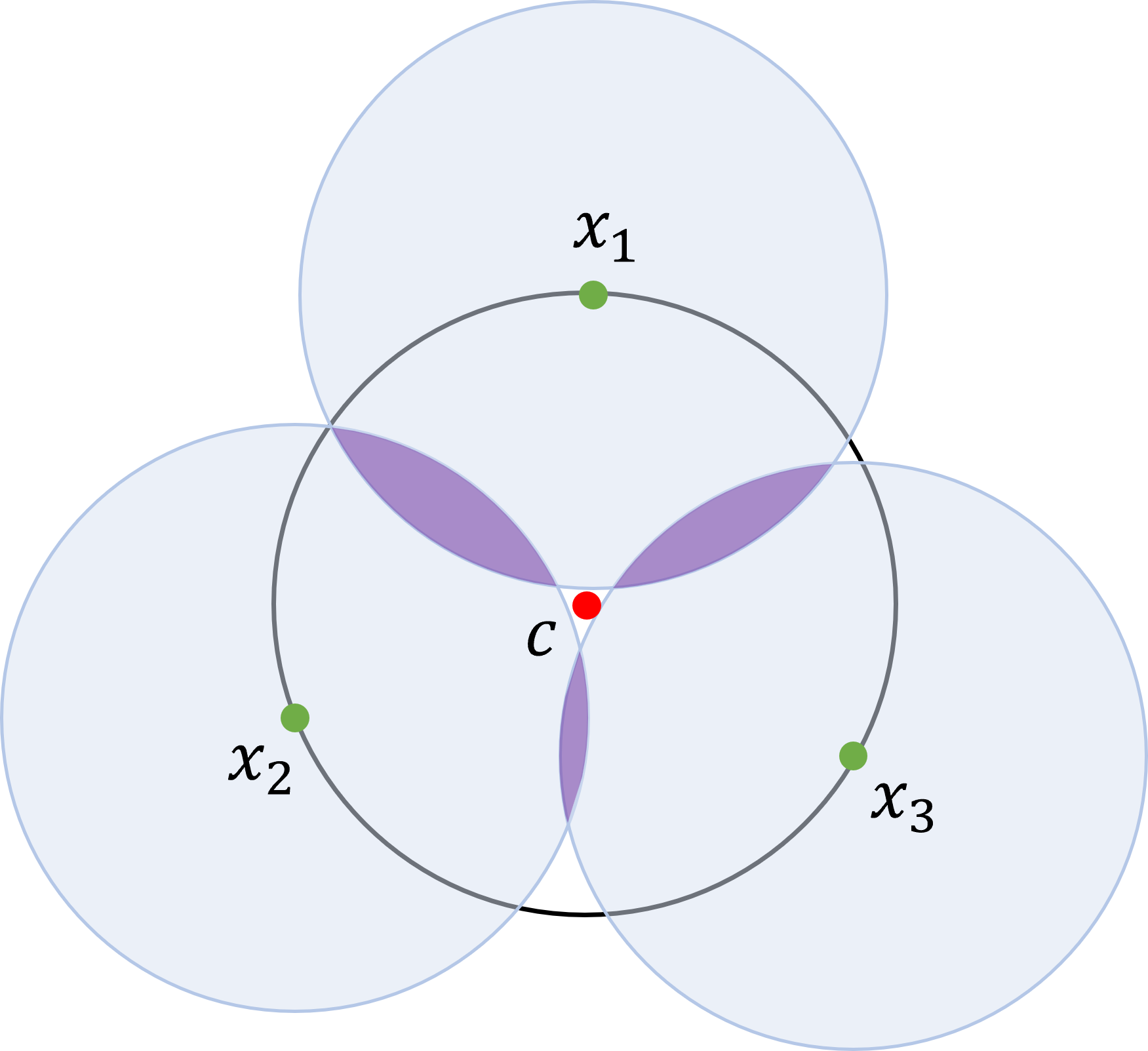}
         \label{fig:three sin x}
     \end{subfigure}
     \vfill\vfill
     \begin{subfigure}[b]{0.45\textwidth}
         \centering
         \includegraphics[width=0.7\textwidth]{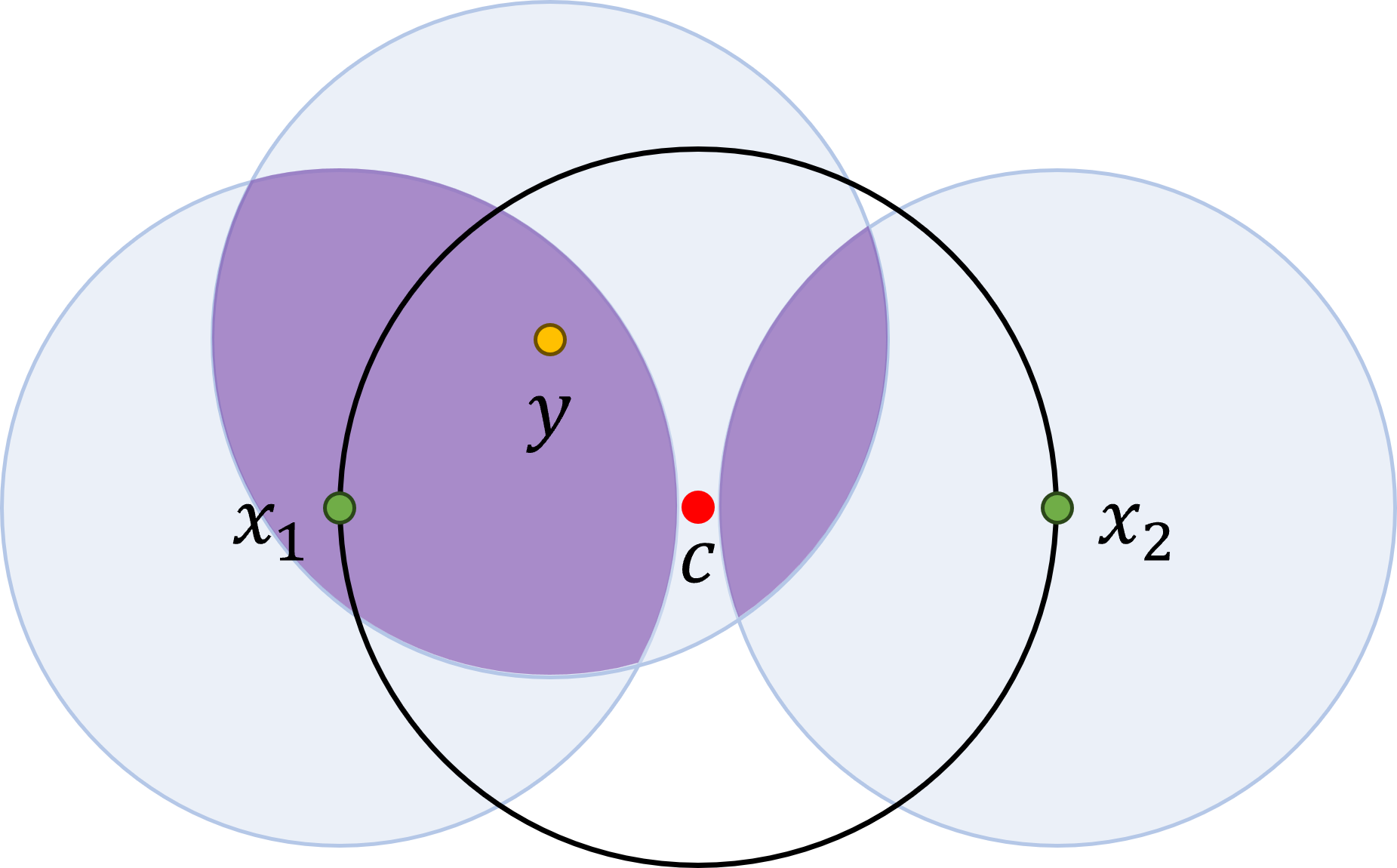}
         \label{fig:five over x}
     \end{subfigure}
     \hfill
     \begin{subfigure}[b]{0.45\textwidth}
         \centering   \includegraphics[width=0.7\textwidth]{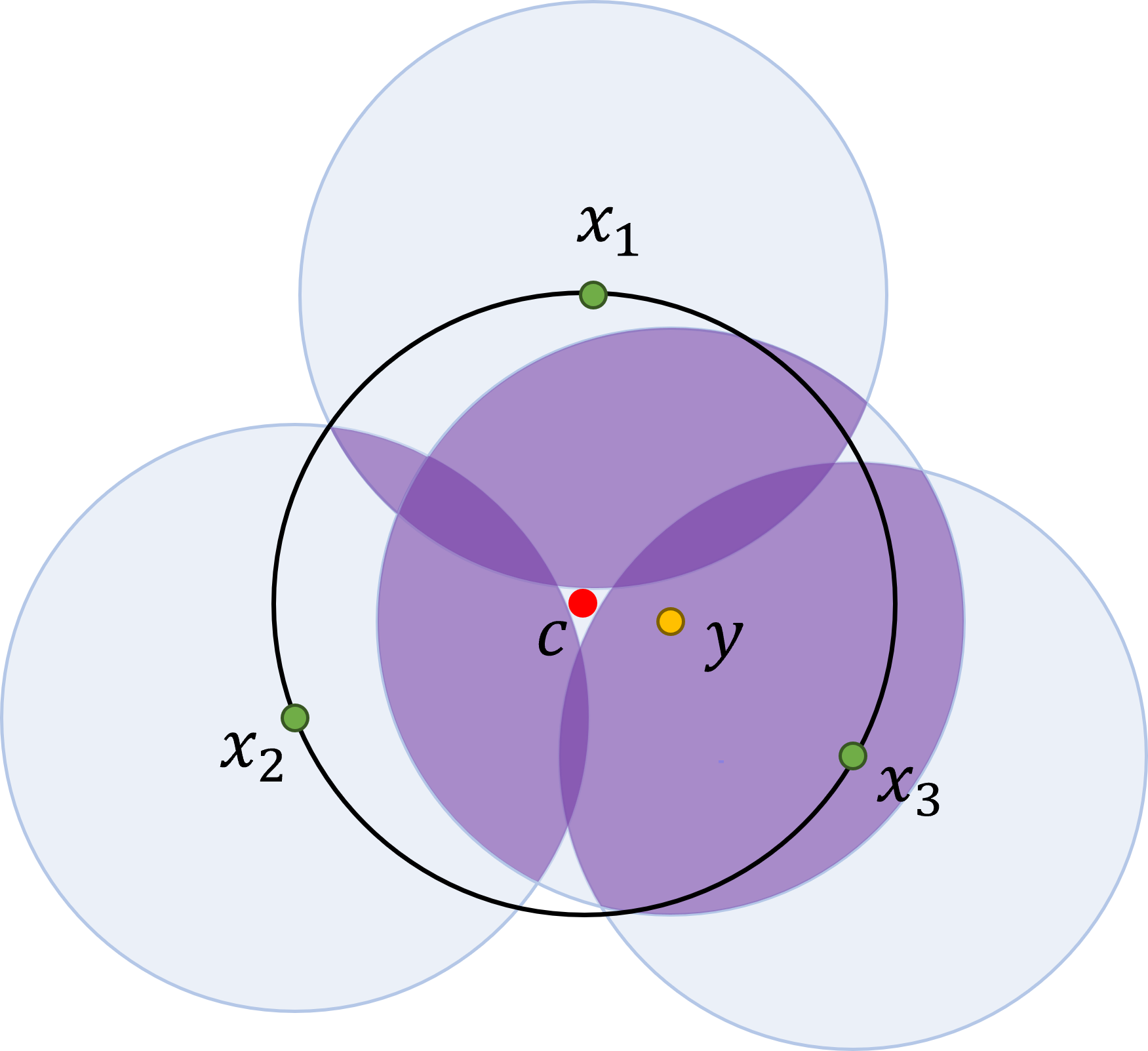}
     \end{subfigure}
    \caption{
    Critical points of $d_{\cP}^{(k)}$ in $\R^2$, for $k=2$. The points $x_1,x_2,x_3$, and $y$ are in $\cP$, and the point $c$ represents the critical point. 
    {\bf Top left:} $\cP_c^{\partial
}=\{x_1,x_2\}$, $\cP_c^{\cI}=\emptyset$, and $\mu_c=0$. This critical point adds a new generator to $H_0$ (new component).
    {\bf Bottom left:} $\cP_c^{\partial
}=\{x_1,x_2\}$, $\cP_c^{\cI}=\{y\}$, and $\mu_c=1$. This critical point kills a generator in $H_0$ (components merge).
    {\bf Top right:} $\cP_c^{\partial
}=\{x_1,x_2,x_3\}$, $\cP_c^{\cI}=\emptyset$, and $\mu_c=1$. This critical point kills two generators in $H_0$ (three components merge into one).
    {\bf Bottom right:} $\cP_c^{\partial
}=\{x_1,x_2,x_3\}$, $\cP_c^{\cI}=\{y\}$, and $\mu_c=2$. This critical point kills an existing $1$-cycle.
    }
    \label{fig:crit_pts}
\end{figure}

The next theorem summarizes the effects of critical points  on the homology of the $k$-fold cover $B_{r}^{(k)}(\cP)$. We consider homology with coefficients in a field $\mathbb{F}$. 
We will assume from here onward that $d_{\cP}^{(k)}$ is a Morse function, in the sense that the critical levels are distinct. While it is easy to find examples where this is not the case, our motivation is the case where $\cP$ is random. For random point-sets, the probability to have two critical points with the same critical value is zero.

\begin{theorem}\label{thm:homology_effect}
Let $c\in\Rd$ be a non-degenerate critical point of $d_{\cP}^{(k)}$ of index $\mu_c$. Let $\epsilon>0$ such that the interval $[r_c-\epsilon,r_c+\epsilon]$ contains a single critical value (namely, $r_c$). Denote $B_r := B_r^{(k)}(\cP)$  and $\Delta_c := \binom{N^{\partial}_c-1}{\mu_c}$.  
Then for $i=\mu_c$, we have
\[
H_{i}(B_{r_c+\eps}) \cong H_{i}(B_{r_c-\eps})\oplus \mathbb{F}^{\Delta_c^+},\quad\text{and}\quad H_{i-1}(B_{r_c-\eps}) \cong H_{i-1}(B_{r_c+\eps})\oplus \mathbb{F}^{\Delta_c^-},
\]where $\Delta_c^+,\Delta_c^-$ are positive integers such that $\Delta_c^+ + \Delta_c^- = \Delta_c$.
If  $i\ne \mu_c,\mu_c-1$, then
\[
H_{i}(B_{r_c+\eps})\cong H_{i}(B_{r_c-\eps}).
\]

\end{theorem}

Note that, this results generalizes the behavior known in classical Morse theory (and for the distance function $d_\cP$), where $\Delta_c = 1$. As in classical Morse theory, finding the values of $\Delta_c^+$ and $\Delta_c^-$ requires additional geometric analysis (see Figure \ref{fig:homology_effect}).
For the highest index ($\mu_c=d$), we have that $\Delta_c = 1$ in the $k$-NN distance as well.

\begin{figure}
    \centering
    \includegraphics[width =0.45\textwidth]{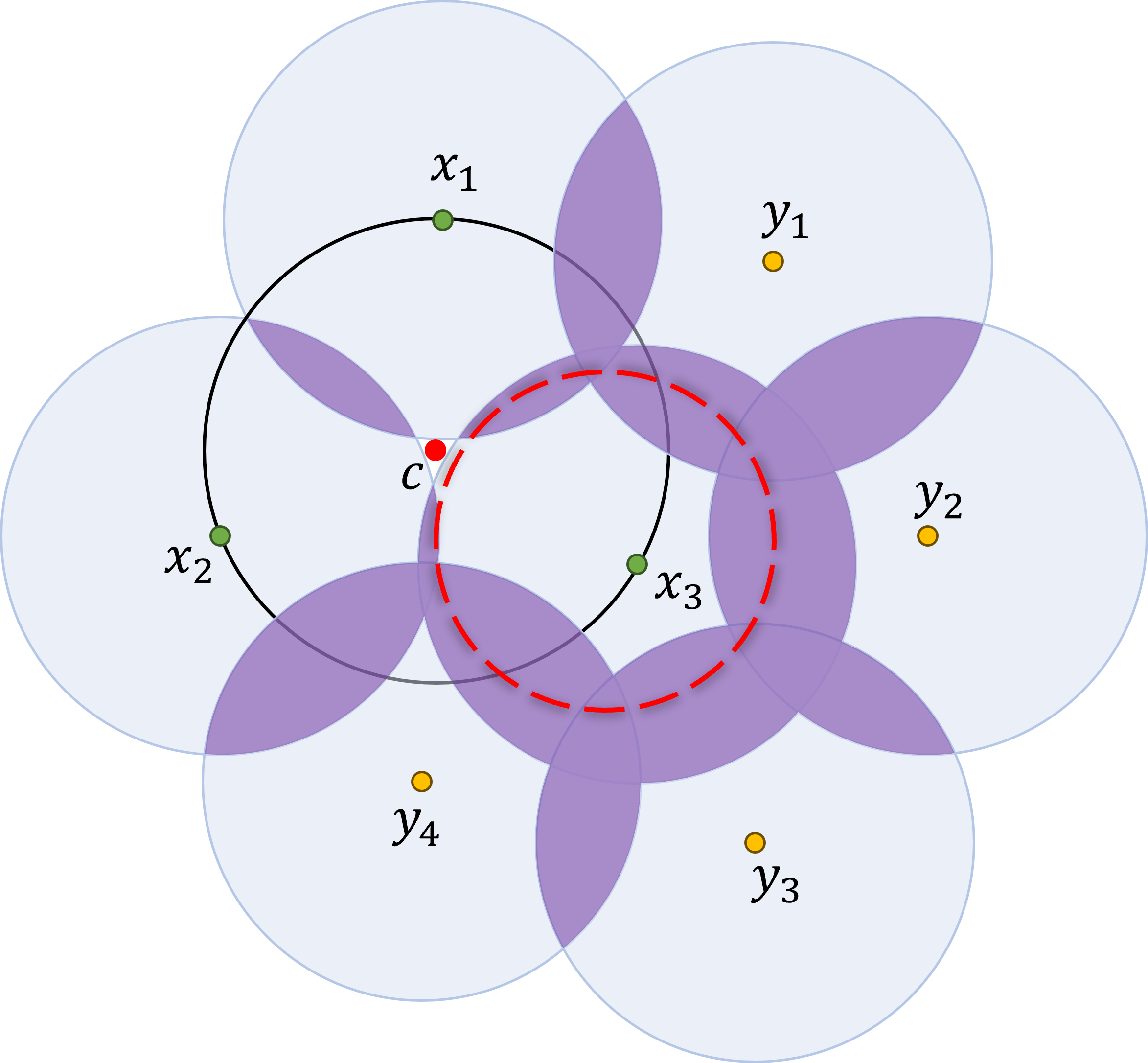}
    \caption{The effect of a critical point on the homology. The point $c\in\R^2$ is a critical point of $d_{\cP}^{(2)}$ of index $\mu_c=1$, where $\cP=\{x_1,x_2,x_3,y_1,\ldots,y_4\}$. In this case, we have $N_c^\partial=3$, and therefore, $\Delta_c = \binom{2}{1} = 2$. Indeed, we observe exactly two changes in the homology of the sub-level sets (purple shaded regions), once $c$ is reached. One change is the generation of a new $1$-cycle on the right side (the red dashed cycle). Another change is the elimination of the connected component ($0$-cycle) on the left side.}
    \label{fig:homology_effect}
\end{figure}

\begin{remark}
While Theorem \ref{thm:homology_effect} provides the total number of changes in the homology, it does not indicate whether these changes are positive (creation of a cycle) or negative (elimination of a cycle). This follows from the nature of Morse theory which is local, while the exact changes are associated with global properties of the ball cover.
One way to identify the exact changes, is via the persistent homology \cite{edelsbrunner2008persistent,zomorodian_computing_2005} of the $k$-Delaunay complex \cite{edelsbrunner2019poisson}.
\end{remark}

\begin{remark}
    The characterization of critical points via Theorem \ref{thm:critical_points} coincides with the notion of `critical steps' in \cite{edelsbrunner2021step}. 
    Thus, an immediate conclusion is that there is a one-to-one correspondence between critical points of $d_{\cP}^{(k)}$ and the critical steps in the order-$k$ Delaunay filtration.
    The approach in \cite{edelsbrunner2021step} was to examine the combinatorial structure of the order-$k$ Delaunay mosaic, and track changes in the Euler characteristic. The Morse-theoretic approach allows us to obtain the detailed description for homology presented in Theorem \ref{thm:homology_effect}.
\end{remark}


\section{Morse theory for piecewise smooth functions}
\label{sec:morse_pw}

We briefly review  the relevant statements from \cite{agrachev1997morse} for piecewise smooth functions. For simplicity, we focus on the case of functions in $\R^d$. 

The definition of critical points in \cite{agrachev1997morse} relies on the notion of the Clarke subdifferential \cite{clarke1975generalized}. Let $f:\R^d\rightarrow\R$ be Lipschitz near a point $x_0\in \R^d$.
The Clarke generalized derivative  at $x_0$ in the direction $v\in \R^d$, is defined as 
\[
        f^{o}(x_0;v) := \limsup_{
        \substack{x\rightarrow x_0 \\ \alpha\rightarrow 0 \\ \alpha>0}
        }
        \frac{f(x + \alpha v) - f(x)}{\alpha}.
\]
The \emph{Clarke subdifferential} of $f$ at $x_0$, denoted by $\partial f(x_0)$, is defined as
\[
    \partial f(x_0) :=\{\xi\in \R^d: f^{o}(x_0;v)\geq \langle\xi, v\rangle\text{ for all }v\in \R^d\}.
\]
This Clarke subdifferential allows us to define critical points for locally-Lipschitz functions.
\begin{definition}
[Definition 1.1 in \cite{agrachev1997morse}]\label{def:crit_pts} Let $f:\R^d\to\R$ be locally Lipschitz. 
    A point $c\in\R^d$ is called critical if ${\bs 0}\in \partial f(c)$.
\end{definition}

Let  $f_1,\ldots, f_m:\R^d\to\R$ be a collection of continuous functions. A  function $f:\R^d\to\R$ is called a \emph{continuous selection} of $f_1,\ldots,f_m$, if for every $x\in\R^d$ we have $f(x) = f_i(x)$ for some $1\le i\le m$.
For every $x$, define
\begin{equation}\label{eqn:I}
    I(x) := \{i: x\in A_i\},\quad A_i := \cl\left(\Int(\set{x: f(x) = f_i(x)})\right),
\end{equation}
where $\cl(\cdot)$ and $\Int(\cdot)$ stand for the closure and interior, respectively.
In the case where $f_1,\ldots,f_m$ are all $C^1$, then $f$ is locally Lipschitz, and its Clarke subdifferential is given by
\begin{equation}\label{eqn:clarke_subdiff}
\partial f(x) = \mathrm{conv}\set{\nabla f_i(x) : i \in I(x)},
\end{equation}
where $\mathrm{conv}$ stands of the convex hull. 
This representation allows us to define non-degenerate critical points for continuous selections.

For fixed $x_0\in \R^d$ and ${\bs{\lambda}} \in \R^{|I(x_0)|}$,  define
\[
L_{\bs{\lambda}}(x):=\sum_{i\in {I}(x_0)}\lambda_i f_{i}(x),\quad\text{and}\quad
{T}(x_0):=\bigcap_{i\in {I}(x_0)} \ker\left(\nabla f_{i}(x_0)\right).
\]

\begin{definition}[Definition 2.2 in \cite{agrachev1997morse}] \label{def:nondegen}
A critical point $c\in\Rd$ is called non-degenerate if the following conditions hold:
\begin{enumerate}
    \item For each $i\in {I}(c)$, the set of gradients $\{\nabla f_j(c):j\in {I}(c){\setminus}\{i\}\}$ is linearly independent,
    
    \item The Hessian of $L_{{\bs{\lambda}}^*}$ at $c$, denoted $H_{{\bs{\lambda}}^*}(c)$, is invertible on $T(c)$ , where ${\bs{\lambda}}^*$ satisfies 
    \[
    \nabla L_{{\bs{\lambda}}^*}(c)=\bs{0},\quad\sum_{i\in {I}(c)}\lambda^*_i = 1,\quad\text{and}\quad\lambda^*_i\geq 0,\text{ for every }i\in {I}(c).
    \]
\end{enumerate}
\end{definition}
Note that \eqref{eqn:clarke_subdiff} guarantees that $\bs{\lambda}^*$ exists, since $\bs{0}$ can be represented as a convex combination of $\{\nabla f_j(c):j\in {I}(c)\}$, and the first condition in Definition \ref{def:nondegen} guarantees that it is unique.
The \emph{quadratic index} $\tilde\mu_c$ is defined as the dimension of the maximal linear subspace of ${T}(c)$ on which $H_{{\bs{\lambda}}^*}(c)$ is negative definite.

According to \cite{agrachev1997morse} (Theorem 2.3), for every non-degenerate point $c$, there exists a neighborhood $U_c$, where $f$ is locally topologically equivalent to a function $g:\R^d\to\R$ of the form
\begin{equation}\label{eqn:local_approx}
g(y) = f(c) + \ell(y_1,\ldots, y_q) - \sum_{j=q+1}^{q+\tilde\mu_c}y_j^2 + \sum_{j=q+\tilde\mu_c+1}^d y_j^2,\quad y = (y_1,\ldots,y_d)\in\tilde U_{\bs 0},
\end{equation}
where $q = |I(c)|-1$, $\ell(y_1,\ldots,y_q)$ is a continuous selection of $\{y_1,\ldots,y_q,-\sum_{j=1}^q y_j\}$,  $\tilde\mu_c$ is the quadratic index, and $\tilde U_{\bs 0}$ is some neighborhood of $\bs 0$.

Next, we define 
\begin{equation}\label{eqn:local_sets}
        U_{c}^{\circ} := \{x\in U_c :  f(x) < f(c)\},
        \quad\text{and}\quad
        U_{c}^{\bullet} := \{x\in U_c: f(x)\leq f(c)\}.
\end{equation}
The following theorem presents the effect of a critical point $c$ on the relative homology.

\begin{theorem}[Theorem 4.2 in \cite{agrachev1997morse}]\label{thm:agrachev_4.2}  
Let $f:\Rd\rightarrow\R$ be locally Lipschitz, and let $c\in\Rd$ be a non-degenerate critical point of $f$. Then,
\begin{enumerate}
\item If $c$ is a local minimum ($U_{c}^{\circ} = \emptyset$), then
    \[
    H_i(U_{c}^{\bullet}, U_{c}^{\circ}) \cong 
    H_i(U_{c}^{\bullet}) \cong
    \begin{cases}
        \mathbb{F} & i=0, \\
        0 & i>0.
    \end{cases}
    \]

    \item If $c$ is \textbf{not} a local minimum, ($U_{c}^{\circ}\neq\emptyset$), then
    \[
    H_i(U_{c}^{\bullet}, U_{c}^{\circ}) \cong  
    \begin{cases}
        H_{i-1}(U_{c}^{\circ}) & i\geq 2, \\
        \mathbb{F}^{\alpha-1} & i=1, \\
        0, & i=0,
    \end{cases}
    \]
\end{enumerate}
where $\alpha$ is the number of connected components of the set $U_{c}^{\circ}$ (with $\mathbb{F}^0\equiv 0$).
\end{theorem}

\section{Critical points for the \headermath{k}-NN distance function}
\label{sec:crit_pts}

Our goal in this section is to prove Theorem \ref{thm:critical_points}.

To simplify some of the calculations, we will prove Theorem \ref{thm:critical_points} for the \emph{squared} $k$-NN distance, denoted $\delta_{\cP}^{(k)} := (d_{\cP}^{(k)})^2$. Any conclusion we make using Morse theory for  $(\delta_{\cP}^{(k)})^{-1}((-\infty,t])$ can be immediately translated to an equivalent statement about $(d_{\cP}^{(k)})^{-1}((-\infty,\sqrt{t}])$. We will therefore consider every critical point of $\delta_{\cP}^{(k)}$ as a critical point of $d_{\cP}^{(k)}$.

Note that to prove Theorem \ref{thm:critical_points} we have to show that the point in question is (a) critical, and (b) non-degenerate. We start with criticality.

\begin{lemma}\label{lem:crit}
A point $c\in\Rd$ is a critical point of $\delta_{\cP}^{(k)}$ if and only if $c\in\sigma(\cP_c^{\partial})$, where $\cP_c^{\partial}$ was defined in \eqref{eqn:def_sets}.
\end{lemma}

\begin{proof}
Without loss of generality we take $c=\bs{0}$. Recall the definition of $I(x)$ in \eqref{eqn:I}, and note that for $\delta_{\cP}^{(k)}$ the indexes in the set $I({\bs 0})$ correspond to the points in $\cP^{\partial}_{\bs 0}$. From \eqref{eqn:clarke_subdiff}, 
we have that
the Clarke subdifferential of $\delta_{\cP}^{(k)}$ at $\bs{0}$ is given by 
\[
\partial \delta_{\cP}^{(k)}(\bs{0}) = \mathrm{conv}(\{\nabla d_{p}^2(\bs{0}):p\in\cP^{\partial}_{\bs{0}}\}).
\] Since $d_p^2$ is the squared distance from $p$, we have $\nabla d_p^2({\bs{0}}) = -2p$, and therefore $\partial \delta_{\cP}^{(k)}(\bs{0}) = -2 \sigma(\cP^{\partial}_{\bs 0})$. 
Since ${\bs 0}\in \sigma(\cP^{\partial}_{\bs{0}})$ if and only if ${\bs 0}\in -2 \sigma(\cP^{\partial}_{\bs{0}})$ (reflected and scaled versions of the same simplex), and using Definition \ref{def:crit_pts}, the proof is  complete. 
\end{proof}

Next, we will show that all critical points in Theorem \ref{thm:critical_points} are indeed non-degenerate.

\begin{lemma}\label{lem:nondegen}
Let $c\in\Rd$, such that $c\in\sigma(\cP_c^{\partial})$. Then, $c$ is non-degenerate for $\delta_{\cP}^{(k)}$.
\end{lemma}

\begin{proof}
As before, we take $c={\bs 0}$.
In our setting we have $f_i = d_{p_i}^2$, and $\nabla f_i({\bs 0}) = -2p_i$. Since we assume the points are in general position, the first condition in Definition \ref{def:nondegen} holds immediately. The Hessian of $f_i$ is $H_i = 2I_{d\times d}$ (the identity matrix). Therefore, $H_{{\bs \lambda}^*} = 2I_{d\times d}$ everywhere, implying that the second condition in Definition \ref{def:nondegen} holds as well.
\end{proof}

\begin{remark}\label{rem:quad_index}
 The proof above shows that the Hessian is always positive definite, and therefore the quadratic index $\tilde\mu_c$ (see Section \ref{sec:morse_pw}) in this case is zero. 
\end{remark}

\begin{proof}[Proof for Theorem \ref{thm:critical_points}]
Follows immediately from Lemma \ref{lem:crit} and \ref{lem:nondegen}.
\end{proof}

\section{Critical points and homology} \label{sec:crit_hom}

In this section, we study the effect of the critical points of $d_{\cP}^{(k)}$ on the homology of its sub-level sets $B_r^{(k)}(\cP)$, and prove Theorem \ref{thm:homology_effect}.
Recall the definition of $U_c^\bullet,U_c^\circ$  in \eqref{eqn:local_sets}. A key observation in the special case of $\delta_{\cP}^{(k)}$ is that the homology of $U_c^\circ$ is simple to describe.

\begin{lemma}\label{lem:Uc_open}
Let $c$ be a non-degenerate critical point of $\delta_{\cP}^{(k)}$, of index $\mu_c$, and denote $\Delta_c := \binom{N^{\partial}_c-1}{\mu_c}$.
If $\mu_c>1$, then
    \[
        H_i(U_c^\circ)\cong \begin{cases} \mathbb{F} & i = 0,\\ \mathbb{F}^{\Delta_c} & i=\mu_c-1,\\
        0 &  \text{otherwise}.
        \end{cases}  
    \]
If $\mu_c=1$, then
    \[
        H_i(U_c^\circ)\cong \begin{cases} \mathbb{F}^{\Delta_c+1} & i = 0,\\ 0 &  \text{otherwise}.
        \end{cases}
    \]
\end{lemma}

The proof for Lemma \ref{lem:Uc_open} requires more  details from \cite{agrachev1997morse}, and is postponed to Section \ref{sec:add_proofs}. We use it here to prove the following simpler version of Theorem \ref{thm:agrachev_4.2}.

\begin{proposition}\label{prop:relative}  
Let $c\in\Rd$ be a critical point of $\delta_{\cP}^{(k)}$ of index $\mu_c$. Then, the following holds.
\begin{enumerate}
    \item If $\mu_c=0$, then $c$ is a local minimum, and 
    \[
    H_i(U_{c}^{\bullet},U_{c}^{\circ}) \cong \begin{cases}
        \mathbb{F} & i = 0,\\
        0 & i > 0.
    \end{cases}     
    \]
    \item If $\mu_c>0$, then 
    \[
    H_{i}(U_{c}^{\bullet},U_{c}^{\circ}) \cong 
    \begin{cases}
    \mathbb{F}^{\Delta_c} & i=\mu_c, \\
    0 & \text{otherwise}.
    \end{cases}
    \]
\end{enumerate}
\end{proposition}

\begin{proof}
When $\mu_c = 0$ we have that $|\cP_c| = k$. Additionally, from Corollary \ref{cor:local_cs}, there is a small neighborhood $U_c$, such that $\delta_{\cP}^{(k)}(x)= \delta_{\cP_c}^{(k)}(x) = \max_{p\in\cP_c}d_p^2(x)$ for all $x\in U_c$.
Note that at $c$ we have $\delta_{\cP}^{(k)}(c) = r_c^2 = d_p^2(c)$, for any $p\in\cP_c^{\partial}$. However, since $c\in \sigma(\cP_c^{\partial})$, for every point $x\in U_c$ there exists $p\in\cP_c^{\partial}$ such that $d_p(x)\ge d_p(c)$. Thus, $c$ is a local minimum.
The first part of the theorem then follows from the first part of Theorem \ref{thm:agrachev_4.2}. 

For $\mu_c>0$, $c$ is not a minimum,  so we refer to the second part of Theorem \ref{thm:agrachev_4.2}.
If $\mu_c>1$, then the result follows directly from Lemma 
\ref{lem:Uc_open}. For $\mu_c=1$, from Lemma \ref{lem:Uc_open}
the number of connected components of $U_c^{\circ}$ is 
$\alpha=\Delta_c + 1 = N_c^{\partial}$, and therefore, the second case in Definition \ref{thm:agrachev_4.2} reduces to $H_{i}(U_c^{\bullet},U_c^{\circ})\cong\mathbb{F}^{\Delta_c}$, for $i= 1$, and $0$ otherwise. 
\end{proof}

We can now prove Theorem \ref{thm:homology_effect}.

\begin{proof}[Proof of Theorem \ref{thm:homology_effect}]
Define
\[
B_r := \{x\in\R^d : d_{\cP}^{(k)}(x) \le r\},\quad\text{and}\quad B_r^\circ := \{x\in\R^d : d_{\cP}^{(k)}(x) < r\}.
\]
Then $U_{c}^{\circ}\simeq U_{c}^{\bullet}{\setminus}\{c\}$, and $B_{r_c}^{\circ}\simeq B_{r_c}{\setminus}\{c\}$, since the critical values of $d_{\cP}^{(k)}$ are distinct. 
By the excision theorem (cf.~Theorem 2.20 in \cite{hatcher_algebraic_2002}), we have 
\begin{equation}\label{eq:excision}
H_i(B_{r_c},B_{r_c}^{\circ})\cong H_i(U_{c}^{\bullet},U_{c}^{\circ}).
\end{equation}

Next, consider the  {long exact sequence} for the relative homology,
\begin{equation*}
\cdots\longrightarrow
H_{i+1}(B_{r_c}, B_{r_c}^{\circ}){\longrightarrow}
H_{i}(B_{r_c}^{\circ}) {\longrightarrow} H_{i}(B_{r_c}) {\longrightarrow}  H_{i}(B_{r_c}, B_{r_c}^{\circ}) 
\longrightarrow\cdots
\end{equation*}
Firstly, consider the case where $\mu_c>0$. Then for $i\ne \mu_c,\mu_c-1 $, from Proposition \ref{prop:relative} we have
\[
0 \longrightarrow
H_{i}(B_{r_c}^{\circ}) {\longrightarrow} H_{i}(B_{r_c}) \longrightarrow 0,
\]
which implies $H_{i}(B_{r_c}^{\circ}) \cong H_{i}(B_{r_c})$. In other words, there is no change in the $i$-th homology of $B_r$ when reaching the point $c$.
For $i=\mu_c$, we have
\[
0 \longrightarrow
H_{i}(B_{r_c}^{\circ}) 
{\longrightarrow} H_{i}(B_{r_c}) {\longrightarrow} 
\mathbb{F}^{\Delta_c}
{\longrightarrow} H_{i-1}(B_{r_c}^{\circ}) 
{\longrightarrow} H_{i-1}(B_{r_c}) \longrightarrow 0.
\]
Exactness then implies that 
\[
H_{i}(B_{r_c}) \cong H_{i}(B_{r_c}^{\circ})\oplus \mathbb{F}^{\Delta_c^+},\quad\text{and}\quad H_{i-1}(B_{r_c}^{\circ}) \cong H_{i-1}(B_{r_c})\oplus \mathbb{F}^{\Delta_c^-},
\] 
for some $\Delta_c^+,\Delta_c^-\ge 0$, with $\Delta_c^++\Delta_c^-=\Delta_c$.

Next, assume that $\mu_c=0$. Similarly to the above, for $i>0$ there is no change in the homology. For $i=0$ we have
\[
0 \longrightarrow
H_{0}(B_{r_c}^{\circ}) {\longrightarrow} H_{0}(B_{r_c}){\longrightarrow} \mathbb{F} \longrightarrow 0.
\]
By exactness, we have $H_{0}(B_{r_c}) \cong H_{0}(B_{r_c}^{\circ}) \oplus\mathbb{F}$. 

Finally, note that there exists $\epsilon>0$ such that the interval $[r_c-\eps,r_c+\eps]$ contains exactly one critical value (namely, $r_c$). From Proposition 2.1 in \cite{agrachev1997morse}, we have that $B_{r_c-\eps}\simeq B_{r_c}^{\circ}$, and $B_{r_c+\eps}\simeq B_{r_c}$. 
This completes the proof.
\end{proof}

\section{Additional proof elements}\label{sec:add_proofs}

In this section we provide more details required for the proofs in Sections \ref{sec:crit_pts} and \ref{sec:crit_hom}.

\subsection{Geometric ingredients}\label{sec:geom_ing}
To use the framework presented in Section \ref{sec:morse_pw}, we will show  that for every $c\in\Rd$ we can  find a small enough neighborhood, where  $d_{\cP}^{(k)}$ is a \emph{continuous selection} of $\set{d_p : p\in\cP^{\partial}_c}$.

\begin{lemma}\label{lem:cs_ball}
Let $\cP\subset\Rd$ be a finite set, and let $c\in\Rd$. Then there exists an open ball $U_c$, centered at $c$, where $d_{\cP}^{(k)}(x)$ 
is a continuous selection of $\{d_{p}:p\in\cP^{\partial}_c\}$, and this representation is minimal.
\end{lemma}
\begin{proof}
Let $r_{\mathrm{in}}=\max_{p\in\cP_c^{\cI}}\|p-c\|$, and $r_{\mathrm{out}}=\min_{p\in\cP\setminus \cP_c}\|p-c\|$. Define
\[
\rho_{\mathrm{in}} = \frac{r_c - r_{\mathrm{in}}}{2},\quad
\rho_{\mathrm{out}} = \frac{r_{\mathrm{out}} - r_c}{2},\quad\text{and}\quad
\rho = \min\{\rho_{\mathrm{in}}, \rho_{\mathrm{out}}\}.
\]
Let $z\in B_{\rho}(c)$. Then the open ball of radius $\frac{r_c+r_\mathrm{in}}{2}$ centered at $z$ includes $B_{r_{\mathrm{in}}}(c)$, and thus all the points in $\cP^{\cI}_c$. In addition, this open ball is included in $B_{r_c}(c)$, and therefore it excludes the points of $\cP\setminus\cP^{\cI}_c$. 
Similarly, the open ball of radius $\frac{r_c+r_{\mathrm{out}}}{2}$ centered at $z$, includes $\cP^{\partial}_c$ and excludes the points of $\cP\setminus\cP_c$.
Thus, setting $U_c=B_{\rho}(c)$ concludes the first part of the proof. 

Next, let $p\in\cP^{\partial}_c$ and denote $\hat{p}_c=(p-c)/\|p-c\|$. Let $\epsilon>0$ sufficiently small, and denote $z=c+\epsilon\hat{p}_c\in U_c$. 
Then, $p$ is necessarily one of the $k$-nearest neighbors of $z$, since for all $q\in\cP^{\partial}_c{\setminus} \{p\}$, we have
\begin{equation*}
\begin{split}
\|q-z\|^2 & = r_c^2 + \epsilon^2 -2\langle q-c, \epsilon\hat{p}_c\rangle
>
r_c^2 + \epsilon^2 -2r_c\epsilon 
=
\|p-z\|^2.
\end{split}
\end{equation*}
Thus, the representation of $d_{\cP}^{(k)}$ as a continuous selection of $\{d_p:p\in\cP^{\partial}_c\}$ is minimal. 
\end{proof}

The $k$-NN distance function is tightly related to the order-$k$ Voronoi tessellation \cite{edelsbrunner1985voronoi,lee1982knnVoronoi}. This is a generalization of  the (order-$1$) Voronoi tessellation, that decomposes $\Rd$ into convex regions whose points have the same $k$-nearest-neighbors.
Formally, let $\cX\subset\cP$ be a subset of size $k$. Then the order-$k$ Voronoi cell of $\cX$ is defined as
\[
\vor(\cX,\cP) := \{y\in \Rd:\|x-y\|\leq\|x'-y\|,\text{ for all }x\in\cX\text{ and }x'\in\cP{\setminus}\cX\}.
\]
Alternatively, we can write
\[
\vor(\cX,\cP) = 
\bigcap_{x\in\cX} \vor(x,\cP{\setminus}\cX),
\]
where $\vor(x,\cP{\setminus}\cX)$ is a standard Voronoi cell.
Note that $\vor(\cX,\cP)$ is a convex set, and  can also be empty. If $\vor(\cX,\cP)\ne\emptyset$ we say that $\cX$ is a $k$-NN subset.

Let $\cP = \set{p_1,\ldots,p_n}\subset\Rd$, and denote all the $k$-NN subsets of $\cP$ by $\cP_1,\ldots,\cP_J$.
In addition, for any $c\in\Rd$, define 
$\Phi_{c}=\{1\leq j\leq J:c\in\vor(\cP_j,\cP)\}$, and for all $j\in\Phi_c$, denote by $\cN_j$ the set of indices, such that $\cP_j\cap\cP^{\partial}_c=\{p_i:i\in\cN_j\}$. 
Using the definitions above, we can write $d_{\cP}^{(k)}$ as 
\begin{equation}\label{eq:min_max}
d_\cP^{(k)} (x) = \min_{1\leq j\leq J} \max_{p\in\cP_j} d_{p}(x).    
\end{equation}
We can refine this representation, using Lemma \ref{lem:cs_ball}.

\begin{corollary}\label{cor:local_cs}
Let $\cP\subset\Rd$ be a finite set, and let $c\in\Rd$. Then there exists a neighborhood $U_c$ where
\[
d_{\cP}^{(k)}(x) = \min_{j\in\Phi_{c}}\max_{p\in\cP_j} d_p(x)
=
\min_{j\in\Phi_{c}}\max_{i\in\cN_j} d_{p_i}(x)
.
\]
In particular, in $U_c$, we have $d_{\cP}^{(k)} \equiv d_{\cP_c}^{(k)}$. 
\end{corollary}

\subsection{Topological ingredients} \label{sec:hom_ing}

Our goal here is to provide a  refined local description for  $\delta_{\cP}^{(k)}$, which will lead to the proof of Lemma \ref{lem:Uc_open}.

Let $c\in\Rd$ be a critical point of $\delta_{\cP}^{(k)}$. 
By Lemma \ref{lem:cs_ball}, there exists a neighborhood $U_c\subset\Rd$, in which
$\delta_{\cP}^{(k)}$ is a continuous selection of $\{d_{p}^2:p\in\cP^{\partial}_c\}$. Moreover, from \eqref{eqn:local_approx}, and since the quadratic index is $0$, it is locally topologically equivalent to 
\begin{equation}\label{eq:top_eq}
g(y)=\delta_{\cP}^{(k)}(c) + \ell(y_1,\ldots,y_{N_c^{\partial}-1}) + \sum_{j=N_c^{\partial}}^{d}y_j^2,\quad y=(y_1,\ldots,y_d)\in \tilde U_{\bs 0}.
\end{equation}
Furthermore, the function $\ell(y_1,\ldots,y_{N_{c}^{\partial}-1})$ admits a min-max representation as a continuous selection of linear functions \cite{agrachev1997morse}. The exact representation is given by the following lemma.

\begin{lemma}
The min-max representation of $\ell(y_1,\ldots,y_{N_{c}^{\partial}-1})$ is given by 
\begin{equation}\label{eq:l_minmax}
\ell(y_1,\ldots,y_{N_{c}^{\partial}-1}) 
=
\min_{j\in\Phi_{c}}\max_{i\in \cN_j} \ell_{i}(y),    
\end{equation}
where $\ell_i(y)$ is one of the functions $y\rightarrow y_l$ ($1\le l\le N_c^\partial-1$), or $y\rightarrow -\sum_ {l=1}^{N_{c}^{\partial}-1}y_l$.    
\end{lemma}

\begin{proof}
Denote $f=\delta_{\cP}^{(k)}$, and assume without loss of generality that $c=\bs{0}$. In addition, denote $q=N_{c}^{\partial}-1$, $H=\R^{q}\times\{0\}^{d-q}$, and $H^{\perp}=\{0\}^{q}\times\R^{d-q}$. Recall that $\cP_{\bs 0}^\partial$ all lie on a $q$-dimensional plane, and assume without loss of generality, that this plane is $H$. 

Let $x\in\Rd$, such that $\|x\|$ is sufficiently small.
Using Corollary \ref{cor:local_cs}, we know that $f(x)$ is determined by one of the functions $\{d_p^2:p\in\cP^{\partial}_c\}$ at $\bs{0}$. We can approximate $f(x)$ based on the second order approximations of $d_p^2$ around $\bs 0$. Namely, 
\begin{equation}\label{eq:f_approx}
f(x)\approx f({\bs 0}) + \min_{j\in\Phi_{c}}\max_{i\in \cN_j} \left(\ell_i(x) + \|x\|^2\right),
\end{equation}
where $\ell_i(x):= \langle\nabla d_{p_i}^2(\bs{0}), x\rangle$, and we used the fact that the Hessian of $d_p^2$ is $2I_{d\times d}$.
Let $x=x^{_\parallel}+x^{_\perp}$, where $x^{_\parallel},x^{_\perp}$ denote the projections of $x$ to $H,H^{\perp}$, respectively. Note that since $\nabla d_{p}^2(\bs{0})\in H$, for all $p\in\cP^{\partial}_c$, we have $\langle\nabla d_{p}^2(\bs{0}), x\rangle = \langle\nabla d_{p}^2(\bs{0}), x^{_\parallel}\rangle$  for all $x\in\Rd$. In other words, the linear terms in \eqref{eq:f_approx} depend only on the first $q=N_c^\partial-1$ coordinates of $x$. 
Thus, for $x^{_\parallel}$, we can assume that the second order term is negligible, while for $x^{_\perp}$ the first order term vanishes. Therefore, we have
\begin{equation}\label{eq:f_approx_final}
f(x)\approx f({\bs 0})+\min_{j\in\Phi_{c}}\max_{i\in \cN_j} \ell_i(x^{_\parallel}) +\|x^{_\perp}\|^2.
\end{equation}
Finally, since the gradients $\{\nabla d^{2}_{\cP}(\bs{0}):p\in\cP_{c}^{\partial}\}$ are linearly dependent, we can express $\ell_{N_c^{\partial}}(x^{_\parallel})$ as $\ell_{N_c^{\partial}}(x^{_\parallel}) = -\sum_{i=1}^{N_{c}^{\partial}-1} \eta_{i} \ell_i(x^{_\parallel})$, where $\eta_{i} := \lambda_{i}^{*}/\lambda_{N_{c}^{\partial}}^{*}$, where $\bs{\lambda}^{*}=(\lambda_{1}^{*},\ldots,\lambda_{N_{c}^{\partial}}^{*})$ is defined in Definition \ref{def:nondegen}. 
Therefore, the form \eqref{eq:f_approx_final} is the same as \eqref{eq:top_eq} up to a change of coordinates, and we can identify $\ell(y_1,\ldots,y_{N_{c}^{\partial}-1})$ with the min-max term in \eqref{eq:f_approx_final}.
\end{proof}

\begin{proof}[Proof of Lemma \ref{lem:Uc_open}]
Based on the min-max representation for $\delta_{\cP}^{(k)}$ we obtained in \eqref{eq:top_eq} and \eqref{eq:l_minmax}, we can use
Theorem 4.1 in \cite{agrachev1997morse} to establish the homology of $U_c^\circ$. This theorem makes use of  an `auxiliary complex', which we compute below for the special case of $\delta_{\cP}^{(k)}$.

Take a critical point $c\in\R^d$ and assume without loss of generality that $\cP_{c}^{\partial}=\{p_1,\ldots,p_{N_c^{\partial}}\}$. For $\ell(y)$ in \eqref{eq:top_eq}, define 
$
S=\{y\in \tilde U_{\bs{0}}: \ell(y)<0 \}. 
$
Following \eqref{eq:l_minmax} we have $S= \bigcup_{j\in\Phi_c} S_j,$ 
where $S_j:=\{y\in \tilde{U}_{\bs{0}}:\ell_i(y)<0,i\in\cN_j\}$.
Given this representation, it was shown in Proposition 2.5 in \cite{agrachev1997morse}
that the following simplicial complex is homotopy equivalent to $S$.
For each $\cN_j$ we define its complement by $\bar{\cN}_j:=\{1,\ldots,N_c^{\partial}\}{\setminus}\cN_j$.
The \emph{auxiliary complex} of $c$, denoted $K_c$, is the nerve of the simplexes $\{\bar{\cN}_j: j\in\Phi_c\}$. 
In fact, for the special case of $\delta_{\cP}^{(k)}$, we observe that $K_c$ is just the $(\mu_c-1)$-dimensional skeleton (recall that $\mu_c=N_c -k$) of the  simplex spanned by $\{1,\ldots,N_c^{\partial}\}$.
Note that the dimension of $K_c$ does not exceed $d$.
See Figure \ref{fig:equi_complex} for examples of this auxiliary complex. 

For the case where the quadratic index is zero (as in our case), Theorem 4.1 in \cite{agarwal_beyond_2005}  states that
$H_i(U_c^\circ) \cong H_i(K_c)$.
Since $K_c$ is the $(\mu_c-1)$ dimensional skeleton of a $(N_c^\partial-1)$-dimensional simplex, we have following. Denote $\Delta_c := \binom{N^{\partial}_c-1}{\mu_c}$. If $\mu_c>1$,
\[
        H_i(K_c)\cong \begin{cases}
        \mathbb{F} & i=0,
        \\
        \mathbb{F}^{\Delta_c} & i=\mu_c-1,
        \\
        0 &  otherwise,
        \end{cases}
\]
If $\mu_c=1$,
\[
        H_i(K_c)\cong \begin{cases}
        \mathbb{F}^{\Delta_c+1} & i=0,
        \\
        0 &  otherwise.
        \end{cases}
\]
This completes the proof.\end{proof}

\begin{figure}[t]
     \centering
     \begin{subfigure}[b]{0.9\textwidth}
         \centering
         \includegraphics[width=0.73\textwidth]{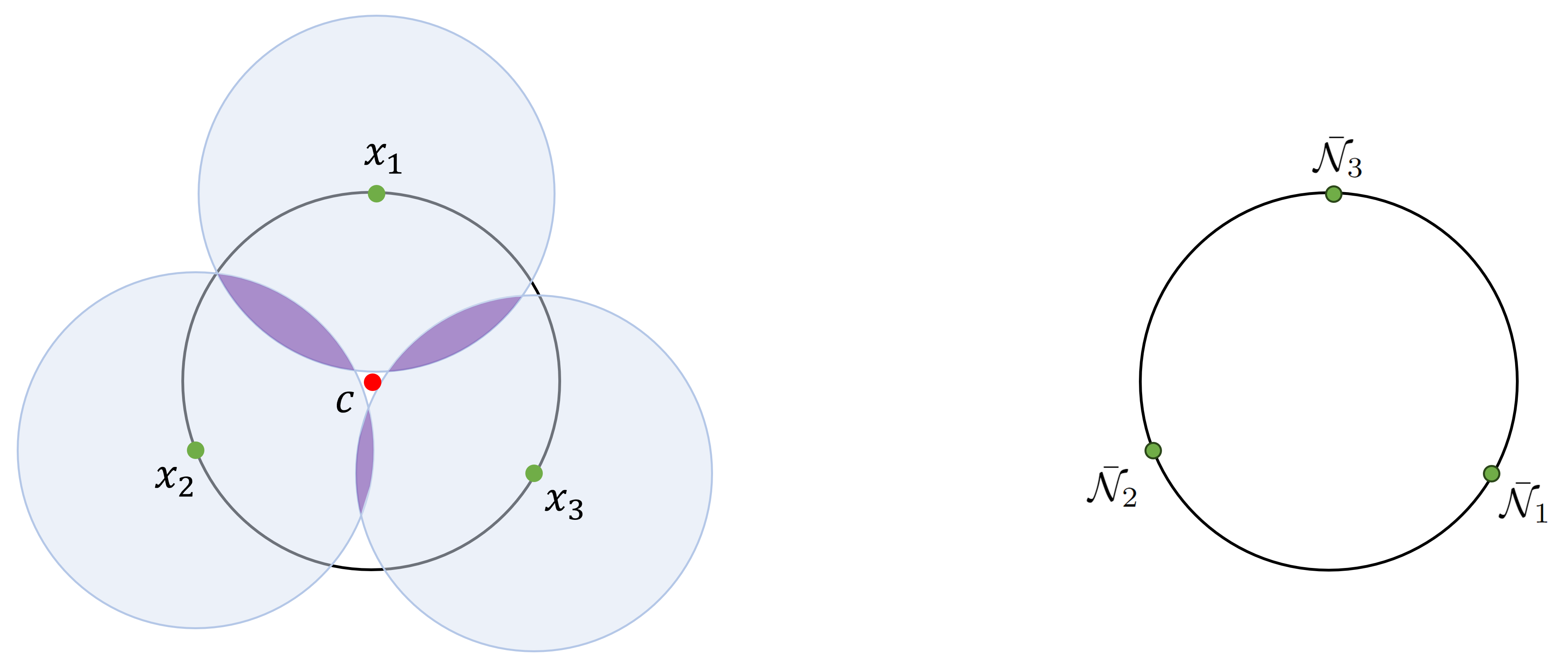}
     \end{subfigure}
     \vfill
     \begin{subfigure}[b]{0.9\textwidth}
         \centering
         \includegraphics[width=0.73\textwidth]{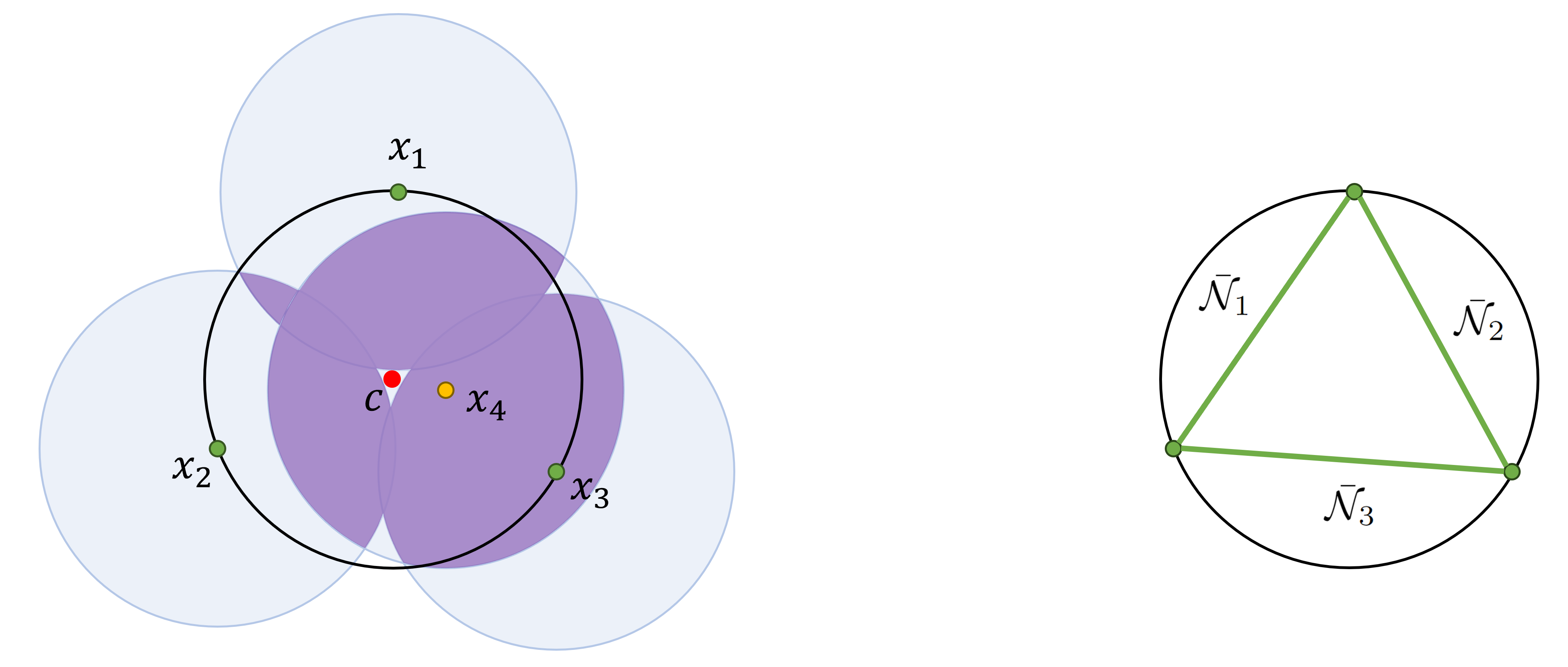}
     \end{subfigure}
        \caption{
        The auxiliary complex used in the proof of Lemma \ref{lem:Uc_open}. 
        {\bf Left:} In both figures $c$ is a critical point of $d_{\cP}^{(2)}$, and the purple regions are the $2$-fold cover, at radius $r$ that is slightly smaller than $r_c$. {\bf Right:} The corresponding auxiliary complex $K_c$ (in green). 
        {\bf Top:} The critical point $c$ is of index $\mu_c=1$. The sets $\cN_1,\cN_2,\cN_3$ are equal to $\{1,2\},\{1,3\},\{2,3\}$, respectively. 
        Thus, the sets $\bar{\cN}_1,\bar{\cN}_2,\bar{\cN}_3$ that span $K_c$, are equal to $\{3\},\{2\},\{1\}$, respectively. 
        (b) The critical point $c$ is of index $\mu_c=2$. The sets $\cN_1,\cN_2,\cN_3$ are equal to $\{1\},\{2\},\{3\}$. Thus, the sets $\bar{\cN}_1,\bar{\cN}_2,\bar{\cN}_3$ that span $K_c$, are equal to $\{2,3\},\{1,3\},\{1,2\}$, respectively.
        }
        \label{fig:equi_complex}
\end{figure}

\section{The expected number of critical points}

In this section we examine the $k$-NN distance function for a random point set $\cP$. In this setting, we analyze the expected value for the number of critical points with a given index, contained within a compact region.

A \emph{homogeneous Poisson point process} in $\Rd$ with intensity $\nu>0$, has the following properties:
\begin{enumerate}
    \item The number of points in a Borel set $A\subset\Rd$ has a Poisson distribution with parameter $\nu|A|$ (where $|\cdot|$ is the volume).
    \item If $A$ and $B$ are two disjoint Borel sets, then the number of points in $A$ and the number of points in $B$ are independent random variables.
\end{enumerate}
The homogeneous Poisson process is a typical case study in stochastic geometry and topology. It has been shown that various topological quantities are linear in $\nu$ (in expectation) \cite{edelsbrunner_expected_2017,edelsbrunner2019poisson,Reani2023Coupled}.
We will show that the critical points for $d_{\cP}^{(k)}$ are no different.

\begin{theorem}\label{thm:exp_crit_pts}
Let $\cP_{\nu}\subset\Rd$ be a homogeneous Poisson point process with intensity $\nu>0$. Let $k>0$, and let $0\leq i \leq d$. Let $\Omega\subset\Rd$ be a compact subset, and denote by $F_{i}$ the number of critical points of $d_{\cP_{\nu}}^{(k)}$, with index $\mu_c=i$, lying in $\Omega$. Then,
\[
\E\{F_{i}\} =D_{k,i}\nu,
\]
where $D_{k,i}$ is a constant that depends on $k$, $i$, and $\Omega$.
\end{theorem}

To prove the above theorem, we  follow the configurations of points in $\cP_\nu$ that generate critical  points for the $k$-NN distance function $d_{\cP_\nu}^{(k)}$.

Let $\cP\subset\Rd$ be a finite set in general position, of size $n>d$. Each critical point of $d_{\cP}^{(k)}$ is associated with a critical configuration of points of $\cP$, as follows. Let $\cX\subset\cP$ of size $l+1$, where $1\leq l\leq d$, and denote by $S(\cX)$ the unique $(l-1)$-dimensional  minimal circumsphere of $\cX$. In addition, denote  
\begin{equation*}
    \begin{split}
        c(\cX) & := \text{the center of } S(\cX), 
        \\
        \rho(\cX) & := \text{the radius of } S(\cX),
        \\
        \cB(\cX) & := \text{the $d$-dimensional ball centered at $c(\cX)$ with radius $\rho(\cX)$}, 
        \\
        \cI(\cX,\cP) & := \mathrm{int}(\cB(\cX))\cap\cP,
        \\
        \mu(\cX,\cP) & := |\cX| + |\cI(\cX,\cP)| - k  
    \end{split}
\end{equation*}
From Theorem  \ref{thm:critical_points} we have that $c=c(\cX)$ is a critical point of $d_{\cP}^{(k)}$ of index $\mu_c:=\mu(\cX,\cP)$, if and only if
\[
0\leq \mu(\cX,\cP)\leq d,\quad\text{and}\quad c\in\sigma(\cX).
\]

\begin{lemma}\label{lemma:F_lm}
Let $\cP_{\nu}\subset\Rd$ be a homogeneous Poisson point process with intensity $\nu>0$. Let $1\leq i\leq d$, and $j\geq 0$. Let $\Omega\subset\Rd$ be a compact subset, and denote by $F_{i,j}$ the number of subsets $\cX\subset\cP_\nu$ of size $|\cX|=i+1$, such that $c(\cX)\in\Omega$, and $|\cI(\cX,\cP)|=j$. Then,
\[
\E\{F_{i,j}\} = D_{d}^{(i,j)}\nu,
\]
where $D_{d}^{(i,j)}$ is a constant that depends on $d$, $i$, $j$, and $\Omega$.
\end{lemma}
For a proof of the above lemma see Appendix \ref{app:A}.

\begin{proof}[Proof of Theorem \ref{thm:exp_crit_pts}]
Recall from Section \ref{sec:main_res} that $\mu(\cX,\cP) = |\cX| + |\cI(\cX,\cP)| - k$ is the index of the generated critical point. In addition, $2\leq |\cX|\leq d+1$, and $0\leq|\cI(\cX,\cP)|\le k-1$.
The last three terms, limit the possible values $|\cX|$ can take, namely 
\[
\max\{2,\mu(\cX,\cP)+1\}\leq|\cX|\leq\min\{d+1,\mu(\cX,\cP)+k\},
\]
and $|\cI(\cX,\cP)| = \mu(\cX,\cP) + k - |\cX|$.
Thus, the number of critical points of index $\mu(\cX,\cP) = i\geq 0$, is given by   
\[
F_{i} = \sum_{i'=I_1}^{I_2}F_{i',j'},
\]
where $I_1:=\max\{1,i\}$,  $I_2:=\min\{d,i+k-1\}$, and $j' = i+k-i'-1$.
By taking the expected value and applying Lemma \ref{lemma:F_lm}, we have
\[
\E\{F_{i}\} = \sum_{i=I_1}^{I_2}\E\{F_{i',j'}\} = \nu \sum_{i'=I_1}^{I_2} D_{d}^{(i',j')}
.
\]
Setting $D_{k,i}:=\sum_{i'=I_1}^{I_2}D_{d}^{(i',j')}$ concludes the proof.
\end{proof}

\section{Discussion}
In this paper we studied the $k$-NN distance function $d_{\cP}^{(k)}$. We showed that using the Morse theory for piecewise smooth functions we can derive simple combinatorial-geometric characterization for critical points and their indices. In addition, we showed the effect of such critical points on the homology of the sub-level sets. We observe that the behavior of $d_{\cP}^{(k)}$ is similar to classical Morse theory, in the sense that if the index is $\mu_c$ the homology affected is only in dimensions $\mu_c$ (positively) and $\mu_c-1$ (negatively). However, in contrast to classical Morse theory, at each critical level there can be several simultaneous changes to homology. 
Our results provide new means to analyze the homology and persistent homology of the $k$-degree Delaunay mosaics. 
In addition, they will be instrumental for the analysis of random $k$-fold coverage and its homology. Specifically, counting critical faces, as we present in Theorem \ref{thm:exp_crit_pts}, will allow us to draw conclusions about the homology of the random $k$-fold coverage objects, in different regimes.
This remains future work.



\newpage
\bibliography{relevant_bib}
\newpage

\appendix

\section{Proof of Lemma \ref{lemma:F_lm}} \label{app:A}
In the following, we provide the proof for Lemma \ref{lemma:F_lm}.
\begin{proof}[Proof of Lemma \ref{lemma:F_lm}]
Fix  $1\le i \le d$, and $j\ge 0$. For finite subsets $\cX\subset \cP\subset \R^d$, with $|\cX| = i+1$, define 
\[
h_{\sigma}(\cX) := \ind\{c(\cX)\in\sigma(\cX)\},\quad 
h_{\cI}(\cX,\cP)=\ind\{|\cI(\cX,\cP)| = j\},
\]
and
\[        
g(\cX,\cP) := h_{\sigma}(\cX)h_{\cI}(\cX,\cP)\ind\{c(\cX)\in\Omega\}. 
\]
Using these notations, we can express $F_{i,j}$ as 
\[
F_{i,j} = \sum_{\substack{\cX\subset\cP_{\nu} \\ |\cX|=i+1}} g(\cX,\cP_{\nu}).
\]
Taking the expectation, and applying the Slivnyac-Mecke formula (see Corollary 3.2.3 in \cite{schneider_stochastic_2008}), yields
\begin{equation}\label{eq:exp_Fij} 
\E\{F_{i,j}\} 
= \frac{\nu^{i+1}}{(i+1)!}		
        \int\limits_{(\Rd)^{i+1}}
		\E
		\{
	g(\mathbf{x}, {{\cP}_{\nu}}\cup \mathbf{x})
		\}
        d\mathbf{x}
        ,
\end{equation}
where abusing notation we treat $\mathbf{x}$ as both an ordered tuple and a set.
For a fixed $\mathbf{x}$ we have
\[
\E\{g(\mathbf{x},{{\cP}_{\nu}}\cup \mathbf{x})\}
=
\frac{\left(\nu\omega_d\rho(\mathbf{x})^d\right)^{j}}{j!}
e^{-\nu\omega_d\rho(\mathbf{x})^d}
h_{\sigma}(\mathbf{x})\ind\{c(\mathbf{x})\in\Omega\}.
\]
Next, we use generalized spherical coordinates (a Blaschke-Petkantschin formula \cite{miles1971isotropic}). Assuming the points in $\mathbf{x}$ are in general position (which is true almost surely), they lie on a unique $i$-dimensional linear space, denoted $\Pi(\mathbf{x})$ (that includes $c(\mathbf{x})$). Recall that the points of $\mathbf{x}$ lie on a $(i-1)$-dimensional sphere centered at $c(\mathbf{x})$ of radius $\rho(\mathbf{x})$. We will denote $\bs{\theta}(\mathbf{x})\subset\Sp^{i-1}$ the spherical coordinates of $\mathbf{x}$ on this sphere. We are interested in the bijective transformation $\mathbf{x}\rightarrow (c,\rho,\Pi,\bs{\theta})$. 

Turning back to the integral in \eqref{eq:exp_Fij}, and applying Lemma C.1 in \cite{bobrowski_homological_2019}, we have
\[ 
\begin{split}
\int\limits_{(\Rd)^{i+1}}
\E
\{
g(\mathbf{x}, &{{\cP}_{\nu}}\cup \mathbf{x})
\}
d\mathbf{x}
\\
&
=
|\Omega| |\mathrm{Gr}(d,i)|
\int\limits_{0}^{\infty}
\int\limits_{(\Sp^{i-1})^{i+1}}
\rho^{di-1} 
\frac{\left(\nu\omega_d \rho^d\right)^{j}}{j!}
e^{-\nu\omega_d \rho^d}
h_{\sigma}(\bs{\theta})
(i!\vsimp(\bs{\theta}))^{d-i+1}
d{\bs{\theta}}d\rho
\\
&
=
C_{d}^{(i,j)} \nu^{j}
\int\limits_{0}^{\infty}
\rho^{d(i+j)-1} 
e^{-\nu\omega_d \rho^d}
d\rho,    
\end{split}
\]
where $\vsimp(\bs{\theta})$ stands for the $i$-dimensional volume of the simplex spanned by $\bs{\theta}$,  
\[
C_{d}^{(i,j)} := \frac{|\Omega||\mathrm{Gr}(d,i)|\omega_d^{j}}{j!}\int\limits_{(\Sp^{i-l})^{i+1}}h_{\sigma}(\bs{\theta})
(i!\vsimp(\bs{\theta}))^{d-i+1} d\bs{\theta},
\]
and $\mathrm{Gr}(d,i)$ is the volume of the $i$-dimensional Grassmannian in $\R^d$.
Taking the  change of variable $t = \nu\omega_d \rho^d$, yields 
\[
\begin{split}
\int\limits_{(\Rd)^{i+1}}
\E
\{
g(\mathbf{x}, {{\cP}_{\nu}}\cup \mathbf{x})
\}
d\mathbf{x}
&
=
\tilde{C}_{d}^{(i,j)} \nu^{-i} 
\int\limits_{0}^{\infty}
t^{i+j-1} 
e^{-t}
dt
=
\tilde{C}_{d}^{(i,j)} \nu^{-i} (i+j-1)!
\end{split}
\]
where $\tilde{C}_{d}^{(i,j)} := \frac{C_{d}^{(i,j)}}{d \omega_d^{i+j}}$. Going back to  \eqref{eq:exp_Fij}, we have 
\[\E\{F_{i,j}\} = D_{d}^{(i,j)} \nu,
\]
where $D_{d}^{(i,j)}:=\frac{\tilde{C}_{d}^{(i,j)} (i+j-1)!  }{(l+1)!}$, concluding the proof.
\end{proof}

\section{Blaschke–Petkantschin-type formula}
The following lemma (Lemma C.1 in \cite{bobrowski_homological_2019}), introduces a change of variables that extends the idea of polar coordinates.
For further details, see \cite{bobrowski_homological_2019}.

Let $\mathbf{x}=(x_1,\ldots,x_{d+1})\subset(\mathbb{T}^d)^{d+1}$, and consider the following mapping $\mathbf{x}\rightarrow(c,\rho,\Pi,\bs{\theta})$ defined in Appendix \ref{app:A}.
Next, let $f:(\mathbb{T}^d)^{k+1}\rightarrow\R$ be affine invariant. This implies that
\begin{equation}\label{eq:aff_inv}
f(\mathbf{x})=f(c+\rho\bs{\theta}(\Pi)) = f(\rho\bs{\theta}(\Pi_0)) := f(\rho\bs{\theta}),    
\end{equation}
where $\Pi_0$ is the canonical embedding of $\R^k$ in $\R^d$ as $\R^k\times\{0\}^k$.
\begin{lemma}[Lemma C.1 in \cite{bobrowski_homological_2019}]\label{lem:C1}
Let $f:(\mathbb{T}^d)^{k+1}\rightarrow\R$ be a measurable bounded function
satisfying \eqref{eq:aff_inv}. Then,
\begin{equation*}
\int_{(\mathbb{T}^d)^{k+1}}f(\mathbf{x})d\mathbf{x} = 
D_{bp}\int_{0}^{\infty}\int_{(\mathbb{S}^{k-1})^{k+1}}\rho^{dk-1}
f(\rho\bs{\theta})(V_{\mathrm{simp}}(\bs{\theta}))^{d-k+1}d\bs{\theta}d\rho,
\end{equation*}
\end{lemma}
where $V_{\mathrm{simp}}(\bs{\theta})$ is the volume of the $k$-simplex spanned by $\bs{\theta}$, $D_{bp}=(k!)^{d-k+1}\Gamma_{d,k}$, and $\Gamma_{d,k}$ is the volume of the Grassmannian $\mathrm{Gr}(d,k)$.

\end{document}